\documentclass[moor]{informs1}   

\usepackage{natbib,ulem}
 \NatBibNumeric
 \def\bibfont{\small}%
 \bibpunct[, ]{[}{]}{,}{n}{}{,}%

\usepackage[colorlinks=true,breaklinks=true,bookmarks=true,urlcolor=blue,
     citecolor=blue,linkcolor=blue,bookmarksopen=false,draft=false]{hyperref}

\def\EMAIL#1{\href{mailto:#1}{#1}}% When hyperref is used, otherwise outcomment 
\TheoremsNumberedThrough     % Preferred (Theorem 1, Lemma 1, Theorem 2)
\EquationsNumberedThrough    % Default: (1), (2), ...
\usepackage{amsmath,amssymb}
\usepackage{algorithm2e}
\usepackage{tikz}
\usepackage{subfig,graphicx}
\usepackage{empheq}

\newcommand{\opt}{{\mathsf{OPT}}}
\newcommand{\algo}{{\mathsf{ALG}}}

\newcommand{\red}[1]{{#1}}

\newcommand{\R}{\mathbb{R}}
\newcommand{\N}{\mathbb{N}}

\def\p{P}
 
\def\E{{\mathbb{E}}}
\def\eps{\varepsilon}
\newcommand{\abc}[3]{\mbox{$\textup{#1}\,|\,\textup{$#2$}\,|\,\textup{$#3$}$}}
\newcommand{\bigO}[1]{\mbox{$\textup{O}(\,#1\,)$}}

\def\ie{{i.e.,\ }}
\def\CV{\mathbb{CV}}
\def\Var{{\mathbb{V}\rm ar}}
\def\eg{{e.\,g.}}

\def\BZ{{\mathbb Z}}

\newcommand{\Prb}[1]{\mbox{$\mathbb{P}[#1]$}}
\newcommand{\incr}[1]{\mbox{$\textup{cost}(#1)$}}
\newcommand{\costub}[1]{\mbox{$\textup{cost}(#1)$}}

\TITLE{Greed Works -- Online Algorithms For Unrelated Machine Stochastic Scheduling}

\ARTICLEAUTHORS{%
\AUTHOR{Varun Gupta}
\AFF{University of Chicago, Chicago (IL), U.S.A.\ \EMAIL{varun.gupta@chicagobooth.edu},}
\AUTHOR{Benjamin Moseley}
\AFF{Carnegie Mellon University, Pittsburgh (PA), U.S.A.\ \EMAIL{moseleyb@andrew.cmu.edu}}
\AUTHOR{Marc Uetz}
\AFF{University of Twente, Enschede, The Netherlands \EMAIL{m.uetz@utwente.nl}}
\AUTHOR{Qiaomin Xie}
\AFF{Massachusetts Institute of Technology, Cambridge (MA), U.S.A.\ \EMAIL{qxie@mit.edu}}
} % e

\ABSTRACT{
This paper establishes performance guarantees for online algorithms that schedule stochastic, nonpreemptive jobs on unrelated machines to minimize the expected total weighted completion time. Prior work on unrelated machine scheduling with stochastic jobs was restricted to the offline case, and required linear or convex programming relaxations for the assignment of jobs to machines. The algorithms introduced in this paper are purely combinatorial. The performance bounds are of the same order of magnitude as those of earlier work, and depend linearly on an upper bound  on the squared coefficient of variation of the jobs' processing times. 
Specifically for deterministic processing times, without and with release times, the competitive ratios are 4 and \red{7.216}, respectively.  As to the technical contribution, the paper shows how dual fitting techniques can be used for stochastic and nonpreemptive scheduling problems.
}

\begin{document}

\maketitle

The point is, ladies and \red{gentlemen}, that greed, for lack of a better word, is good. Greed is right, greed works (Gordon Gekko in {\it Wall Street} \cite{WallSt})

%\footnote{Gordon Gekko (Michael Douglas) in Oliver Stone's ``Wall Street'' (Twentieth Century Fox, 1987)
% !TEX root =  RevisionMOR2017-339.tex

\section{Introduction.}
Scheduling jobs on multiple, parallel machines is a fundamental problem both in combinatorial optimization and systems theory. There is a vast amount of different model variants as well as applications, which is testified by the existence of the handbook \cite{HandbookScheduling}. A well studied class of problems is scheduling a set of $n$ nonpreemptive jobs that arrive over time on $m$ unrelated machines with the objective of  minimizing the total weighted {completion time}. In the unrelated machines model the matrix that describes the processing times of all jobs on all machines can have any rank larger than 1.  The offline version of the problem is denoted \abc{R}{r_{j}}{\sum w_jC_j} in the three-field notation of Graham et al.~\cite{gllrk1979}, and the problem has been a cornerstone problem for the development of new techniques in the design of (approximation) algorithms, e.g.\ \cite{BSS,horowitz1976exact,LST,Skut-JACM01}.

This paper addresses the online version of the problem where jobs sizes are stochastic. In the online model  jobs arrive over time, and the set of jobs is unknown a priori. For pointers to relevant work on online models in scheduling, refer to \cite{ImMP11,PruhsST}. In many systems, the scheduler may not know the exact processing times of jobs when the jobs arrive to the system. Different approaches have been introduced to cope with this uncertainty.  If jobs can be preempted, then non-clairvoyant schedulers have been studied that do not know the processing time of a job until the job is completed \cite{MotwaniPT94,BecchettiL01,KalyanasundaramP00,GuptaIKMP12,ImKMP14}. Unfortunately, if preemption is not allowed then any algorithm has poor performance in the non-clairvoyant model, as the lower bound for the competitive ratio against the offline optimal schedule is $\Omega(n)$. This is even true if we consider the special case where all jobs have the same unit weight $w_j$.

This lower bound suggests that the non-clairvoyant model is too pessimistic for non-preemptive problems. Even if exact processing times are unknown to the scheduler, it can be realistic to assume that at least an estimate of the true processing times is available.  For such systems, a model that is used is \emph{stochastic scheduling}.  In the stochastic scheduling model the jobs' processing times are given by random variables.  A \emph{non-anticipatory} scheduler only knows the random variable  that encodes the possible realizations of a job's processing time. If the scheduler starts a job on a machine, then that job must be run to completion \emph{non-preemptively}, and it is only when the job completes that the scheduler learns the actual processing time of the job. With respect to the random processing times, both the scheduler and the optimal solution are required to be non-anticipatory, which means that only the (conditional) distribution of a job's processing time may be used at any point in time. Stochastic scheduling has been well-studied, including fundamental work such as by M\"ohring et al.\ \cite{mrw1,mrw2} and approximation algorithms, e.g.\ \cite{MSU99,SkutUetz05,muv2006,SSU2016,schulz2008}.

This paper considers online scheduling of non-preemptive, stochastic jobs in the unrelated machine model to minimize the total weighted completion time. This is the same problem as considered in the paper \cite{muv2006} by Megow et al., but here we address the more general \emph{unrelated} machines model.  In the stochastic unrelated machine model, the scheduler is given machine-dependent probability distributions that describe a job's potential processing time for each of the machines.  For a given job the processing times across different machines need not be independent, but the processing times of different jobs are assumed to be independent.

\noindent
{\bf Identical machines, special processing time distributions.} Restricting attention to non-preemptive policies, when all machines are identical, perhaps the most natural algorithm is Weighted Shortest Expected Processing Time (WSEPT) first. When a machine is free, WSEPT always assigns the job to be processed that has the maximum ratio of weight over expected size. When all jobs have unit weight, this algorithm boils down to the SEPT algorithm that  greedily schedules jobs with the smallest expected size.  When there is a single machine and all jobs arrive at the same time, WSEPT is optimal~\cite{Rothkopf66}. For multiple machines with equal weights for all jobs, if the job sizes are deterministic and arrive at the same time, SEPT is optimal~\cite{Horn73}. For multiple identical machines with equal weights for all jobs, SEPT is optimal if job sizes are exponentially distributed \cite{bruno_downey_fred81,wp80}, or more generally,  are stochastically comparable in pairs~\cite{wvw86}.  Some extensions of these optimality results to the problem with weights exist as well \cite{ka87}. For more general distributions, simple solutions fail \cite{uetz03}, and our knowledge of optimal scheduling policies is  limited.

\noindent
{\bf Identical machines, arbitrary processing times.} To cope with these challenges, approximation algorithms have been studied. With the notable exception of \cite{ImMP15}, all approximation algorithms 
have performance guarantees that depend on an upper bound $\Delta$ on the squared coefficient of variation of the underlying random variables. M{\"{o}}hring, Schulz and Uetz \cite{MSU99} established the first approximation algorithms for stochastic scheduling on identical machines via a linear programming relaxation.  Their work gave a $(3+\Delta)$-approximation  when jobs are released over time (yet known offline), and they additionally showed that WSEPT is a $(3+\Delta)/2$-approximation when jobs arrive together\footnote{The ratio is slightly better, but for simplicity we ignore the additive $\Theta(1/m)$ term. }.  These results have been built on and generalized in several settings \cite{SkutUetz05,mv2006,muv2006,schulz2008,SSU2016,jaegerskutella2018}, notably in \cite{muv2006,schulz2008} for the online setting.
The currently best known result when jobs are released over time (yet known offline) is a $(2+\Delta)$-approximation by Schulz~\cite{schulz2008}.  In the online setting Schulz gives an algorithm  with  performance guarantee of $(2.309 + 1.309\Delta)$~\cite{schulz2008}. These results  build on an idea from Correa and Wagner \cite{cw2008} to use a preemptive, fast single machine relaxation, next to the relaxation of \cite{MSU99}.  The work of Im, Moseley and Pruhs~\cite{ImMP15} gave the first results independent of $\Delta$ showing that there exist poly-logarithmic approximation algorithms under some assumptions.  All these papers address problems with identical machines. 

\noindent
{\bf Unrelated machines, arbitrary processing times.} For some 15 years after the results of M\"ohring et al.\ \cite{MSU99} for the identical machines case, no non-trivial results were known for the \emph{unrelated} machines case despite being a  target in the area. Recently  Skutella, Sviridenko and Uetz~\cite{SSU2016}  gave a $(3+\Delta)/2$-approximation algorithm for the unrelated machines model when jobs arrive at the same time, and a $(2+\Delta)$-approximation when jobs are released over time (yet known offline). Central to unlocking an efficient approximation algorithm for the unrelated machines case was the introduction of a  time-indexed linear program that lower bounds the objective value of the optimal non-anticipatory scheduling policy. It is this LP that allows the authors to overcome the complexities of the stochastic unrelated machines setting.

The work introduced in this present paper targets the more realistic \emph{online} setting for  scheduling  stochastic jobs on unrelated machines.  A priori, it is not clear that there should exist an algorithm with small competitive ratio for this problem.  Prior work for the offline problem requires sophisticated linear \cite{SSU2016} or convex \cite{BBC2016} programming relaxations. Good candidates for online algorithms that might also have practical impact are desired to be simple and combinatorial, but even discovering an offline approximation algorithm that is simple and combinatorial has remained an open problem for (stochastic) scheduling on unrelated machines. 

\noindent
{\bf Related work for deterministic processing times.} For special cases and deterministic processing times, approximation algorithms have been known to exist. For example for the online unrelated machine case with deterministic
processing times, Hall, Schulz, Shmoys and Wein~\cite{HSSW97} obtain an $8$-competitive algorithm. Their algorithm is  based on the idea to partition time into geometrically increasing intervals, and then maximizing the total weight of (available) jobs that can be scheduled in these intervals. Algorithms with better competitive ratios have been obtained by Chakrabarti et al.~\cite{chakra96} by using randomization in the definition of these intervals, and resulting in a randomized 5.78-competitive algorithm. As far as we know, this is the state of the art when it comes to competitive analysis
for the online problem with release times and on unrelated machines. 
The deterministic greedy algorithm proposed in this paper is 6-competitive.

For the offline problem with deterministic processing times, the following is known. 
When there are no release times and processing times are deterministic, the currently best known approximation algorithms have performance bounds slightly below $3/2$, based on semidefinite relaxations \cite{BSS} and more recently also on linear relaxations \cite{Li17}. For the offline case with release times, the $(2+\varepsilon)$-approximation of \cite{SS-SIDMA} was the best known until recently Im and Li \cite{ImL16} gave a $1.878$-approximation algorithm.
The problem has also been looked at through the lens of game theory, and for the offline problem without release times, Cole et al.~\cite{ColeCGMO11} show that when machines follows the WSPT rule, Nash equilibria of selfish jobs that minimize their own completion time yield schedules which are at most a factor 4 above optimum. Interestingly, our paper shows that the same approximation guarantee can be obtained online by a simple greedy algorithm.  We note that the work of \cite{ColeCGMO11} is offline, and moreover their algorithm and analysis differs from that of this paper.

With respect to lower bounds on performance guarantees of online algorithms, 
we are aware of only one lower bound on the competitive ratio of any online algorithm, which is the $1.309$ lower bound of Vestjens~\cite{Vestjens97}; this lower bound  holds for the 
problem on identical machines with deterministic processing times.

\noindent
\textbf{Results.}  This paper suggests two {combinatorial} online algorithm for stochastic scheduling on unrelated machines that have a performance guarantee of order $\bigO{\Delta}$, where $\Delta$ is an upper bound on the squared coefficient of variation of the processing time distributions $\p_{ij}$.  More specifically, in the online-list model where jobs arrive online (at time 0) and must be assigned to a machine immediately upon arrival, this paper establishes a performance guarantee of $(4+2\Delta)$. For this problem the algorithm assigns jobs to machines so as to minimize the expected contribution to the objective function, while per machine the jobs are sequenced by largest ratio $w_i/\E[\p_{ij}]$ first. For deterministic processing times, the proposed greedy algorithm has a competitive ratio of 4, and the paper also gives a lower bound instance showing that the analysis is tight. Arguably more relevant is the online-time model where jobs arrive over time at individual release times. Here, the paper establishes a performance guarantee of  \red{$(7.216+3.608\Delta)h(\Delta)$}, where $h(\Delta)  = 1+\sqrt{\Delta}/2$ for  $ \Delta \leq 1$, and  $h(\Delta)  = 1+\Delta/(\Delta+1)$ for $\Delta \geq 1$. Observe that $h(\cdot)$ is a concave, increasing function of $\Delta$ which is bounded from above by $2$. Here, the greedy assignment of jobs to machines is the same, but the sequencing per machine is augmented by possibly introducing forced idle time. The idea is to work with a ``nominal'' schedule based on expected processing times, and never start processing a job before its nominal starting time.
For deterministic processing times $\Delta=0$ and $h(\Delta)=1$, hence the competitive ratio equals \red{$7.216$}. As the algorithm is a deterministic algorithm, this improves upon the competitive ratio $8$ from \cite{HSSW97}, \red{but} falls slightly behind the randomized $5.78$-competitive algorithm that was proposed in \cite{chakra96}.

Even though the performance bounds for the case with nontrivial release times are most probably not tight, we believe our results are interesting for the following reasons:
(1) It is the first analysis of a combinatorial algorithm for stochastic scheduling on unrelated machines, and the first result for stochastic online scheduling in the unrelated machine model. (2) Even for the deterministic setting, it is the first time to analyze an (arguably) intuitive combinatorial algorithm that simply assigns jobs to machines where their expected contribution is minimal.
(3) The analysis uses the idea of dual fitting, hence we demonstrate that this technique can be used for bounding the performance of scheduling policies in non-preemptive and stochastic scheduling. (4) The performance bounds, even where not tight, have the same order of magnitude as those of earlier results in the literature, while considering a more general problem.

We now briefly discuss the proposed algorithms in relation to prior work. The algorithms rest on the following ingredients to solve the machine assignment and the scheduling problem. Generally speaking, at any point in time a job with highest ratio $w_j/\E[\p_{ij}]$ is scheduled from the set of jobs that are assigned to and available on machine~$i$. This is the well known WSEPT rule. For the case with release dates, however, jobs possibly have to wait for ``artificial'' release times before they are declared available for processing. The necessity of such forced idle time is well known whenever jobs are released over time, even for single machine problems \cite{ms2004}. Next to this standard manipulation of release times, we work with a ``nominal'' schedule that is based on expected processing times, and never allow jobs to be started before their nominal starting times.
The assignment of jobs to machines is solved by greedily assigning jobs to the machines where (a proxy for) the expected increase of the objective is minimal. Comparable greedy-type algorithms have been used also before, e.g.\ in \cite{AA2003,ms2004,muv2006}, however not for an unrelated machine setting.
Note that the $\Omega(\Delta)$ lower bound for fixed assignment policies in \cite{SSU2016} yields that our results are asymptotically tight in $\Delta$ among policies that must irrevocably assign jobs to machines at the time of their release.
As mentioned above, the analysis proposed in this paper uses dual fitting techniques.  The technique has been used e.g.\ in  \cite{AnandGK12} for deterministic and preemptive scheduling problems.

% !TEX root =  RevisionMOR2017-339.tex

\section{Notation \& Preliminaries.}

The input to the problem consists of a set of unrelated parallel machines~$M$ of cardinality~$m$. 
The set of jobs $J$, of cardinality $n$, is unknown and only disclosed gradually over time. Each job needs to be executed on exactly one (and any one) of the machines in~$M$, and each machine can process at most one job at a time. The jobs are nonpreemptive. This means that a job, once started, must not be interrupted until its completion.

This paper considers  two online models. In the first model, known as \emph{online-list}, the scheduler is presented the jobs $j\in J$ one after the other. Whenever a job is presented, the algorithm has to assign it to one of the machines before the next job is presented. The decision when the job begins being processed can be deferred until all jobs have arrived. It is unknown how many jobs will arrive, but once all jobs in $J$ have arrived, the jobs assigned to any one of the machines must be sequenced on that machine in some order. In the second model, known as \emph{online-time}, time progresses and jobs appear over time at their individual release times.  Let $r_j$ denote the release time of job $j$.
At the moment of arrival $r_j$, or possibly at a later point in time, a job must be assigned to a machine. Once assigned to a machine, the job may possibly wait until an even later point in time to be processed.

The jobs are stochastic, meaning that each job $j$'s processing time is revealed to the scheduler at the point of arrival in the form of a random variable~$\p_{ij}$ for every machine~$i\in M$.  If job~$j$ is assigned to machine~$i$, its processing time will be random according to~$\p_{ij}$. It is allowed that certain jobs $j\in J$ cannot be processed on certain machines~$i\in M$, in which case $\E[\p_{ij}]=\infty$.

In the stochastic scheduling model, the realization of the processing time of a job~$j$ becomes known at the moment that the job completes. This paper considers designing a non-anticipatory scheduling policy $\Pi$ that  minimizes the expected total weighted completion time~$\E\bigl[\sum_j w_jC_j\bigr]$, where~$C_j$ denotes the random variable for the completion time of job~$j$ under policy $\Pi$.

We assume that the random variables $\p_{ij}$ are discrete and integer valued. This can be assumed at the cost of a multiplicative factor of $(1+\eps)$ in the final approximation ratio, for any $\eps>0$ \cite{SSU2016}.
Our analysis will  make use of the following facts about first and second moments of discrete random variables; these facts also appear in \cite{SSU2016}.
\begin{lemma}\label{lem:moment}
Let~$X$ be an integer-valued, nonnegative random variable. Then,
\begin{align*}
\sum_{r\in\BZ_{\geq0}}\Prb{X>r}&=\E[X]\quad
\text{and}\quad	
\sum_{r\in\BZ_{\geq0}}(r+\tfrac12)\,\Prb{X>r}=\frac12 \E[X^2]\,.
\end{align*}
\end{lemma}

\begin{definition}
Let~$X$ be a nonnegative random variable. The \emph{squared coefficient of variation} is defined as the scaled variance of $X$. That is,
\[
\CV[X]^2:=\Var[X]/\E[X]^2\,,
\]
where $\Var[X]=\E[X^2]-E[X]^2$.
\end{definition}

\subsection{Stochastic Online Scheduling \& Policies}
The setting considered in this paper is that of stochastic online scheduling as defined in \cite{muv2006}.  This means that (the existence of) a job $j$ is unknown before it arrives, and upon arrival at time $r_j$, only the distribution of the random variables $P_{ij}$ for the possible processing times on machines $i=1,\dots,m$ are known to the scheduler. At any given time $t$, a non-anticipatory online scheduling policy is allowed to use only the information that is available at time $t$. In particular, it may anticipate the (so far) realized  processing times of jobs up to time $t$. For example, a job that has possible sizes 1, 3 or 4 with probabilities 1/3 each, and has been running for 2 time units, will have a processing time 3 or 4, each with probability 1/2. It is well known that adaptivity over time is needed in order to minimize the expectation of the total weighted completion time, e.g. \cite{uetz03}. We refer the reader to \cite{muv2006} for a more thorough discussion of the stochastic online model.

For simplicity of notation, denote $\opt$ as the expected total weighted completion time of an optimal, non-anticipatory scheduling policy for the problem where the set of jobs, their release times $r_j$ and their processing time distributions $\p_{ij}$ are known in advance. That is to say, the benchmark that we compare our algorithms to, knows the set of jobs and their parameters, but not the actual realizations of processing time distributions $\p_{ij}$.

We seek to find a non-anticipatory online scheduling policy (an algorithm) $\algo$ with expected performance $\algo$ close to $\opt$. For convenience, and in a slight abuse of notation we use the same notation for both the algorithm and its expected performance. That is to say, both $\algo$ and $\opt$ denote the expected performance of non-anticipatory scheduling policies, and by linearity of expectation we have
$\algo= \sum_{j}w_j\E[C_j^{\algo}]$ and $\opt= \sum_{j}w_j\E[C_j^{\opt}]$.

\begin{definition}
A scheduling policy  is said to have a (multiplicative) \emph{performance guarantee} $\alpha\geq 1$, if for every possible input instance, 
\[
\algo\leq \alpha\opt\,.
\] 
\end{definition}
We remark that $\opt$ is not restricted to assigning jobs to machines at the time of their arrival. The only restriction on $\opt$ is that it must schedule jobs non-preemptively, and that it is non-anticipatory. Note that our approximation guarantees hold against an adversary who knows all the jobs and their release times $r_j$, as well as the processing time distributions $\p_{ij}$ in advance, but not the actual realizations of $P_{ij}$. That implies that the model generalizes the classic offline stochastic scheduling model (assuming all paramaters are disclosed to the scheduler, too), as well as traditional competitive analysis (assuming deterministic processing times).

Finally, we may assume w.l.o.g.\ that no pair of job and machine exists with $\E[\p_{ij}]=0$. That said, we may further assume that $\E[\p_{ij}]\ge 1$ for all machines $i$ and jobs $j$, by scaling.

% !TEX root =  RevisionMOR2017-339.tex

\section{Linear Programming Relaxations.} \label{sec:LP}
This section introduces a linear programming relaxation for the problem.  This relaxation was previously discussed  in \cite[\S 8]{SSU2016}. The LP uses variables~$y_{ijs}$ to denote the probability that job~$j$ is being processed on machine~$i$ within the time interval~$[s,s+1]$, under some given and fixed scheduling policy. It is known that $y_{ijs}$ can be linearly expressed in terms of the variables $x_{ijt}$, which denote the probability that job $j$ is started at time $t$ on machine $i$, as follows
\begin{align}
y_{ijs}=\sum_{t=0}^s x_{ijt}\, \Prb{\p_{ij}>s-t}\enspace.\label{eq:fromxtoy}
\end{align}
The fact that any machine can process at most one job at a time can be written as
\begin{align}
\sum_{j\in J}y_{ijs}\leq 1 \qquad\text{for all $i\in M$, $s\in\BZ_{\geq0}$.}\label{eq:LPy:mach-cap}	
\end{align}
Moreover, by
the fact that scheduling policies are non-anticipatory we know that whenever a job $j$ is started on a machine $i$ at time $t$, it will in expectation be processed for time $\E[\p_{ij}]$, so its expected completion time is $t+\E[\p_{ij}]$. Now, conditioning on a job being processed on machine $i$, 
making use of \eqref{eq:fromxtoy} and the first part of Lemma~\ref{lem:moment}, together with the fact that each job must be completely processed, 
gives the constraint that $\sum_{s\in\BZ_{\geq0}}\frac{y_{ijs}}{\E[\p_{ij}]}=1$. Unconditioning on the machine assignment 
yields the 
following constraints
\begin{align}
\sum_{i\in M}\sum_{s\in\BZ_{\geq0}}\frac{y_{ijs}}{\E[\p_{ij}]}=1 \qquad\text{for all~$j\in J$}\,.\label{eq:LPy:assignment}
\end{align}
Finally, with the help of \eqref{eq:fromxtoy} and the second part of Lemma~\ref{lem:moment}, the expected completion time of a job~$j$ can be expressed in $y_{ijs}$ variables as
\begin{align}
C_j^{\ref{sto}}:=\sum_{i\in M}\sum_{s\in\BZ_{\geq0}}
\left(\frac{y_{ijs}}{\E[\p_{ij}]}\,\bigl(s+\tfrac12\bigr)+\frac{1-\CV[\p_{ij}]^2}{2}\,y_{ijs}\right) \qquad\text{for all~$j\in J$\,,}\label{eq:LPy:compl-time}
\end{align}
where we labeled the expected completion time variables with a superscript \ref{sto} for ``stochastic'', for reasons that will become clear shortly. For completeness, equation \eqref{eq:LPy:compl-time} is proved in Lemma~\ref{lemma:LPy:compl-time} (Appendix~\ref{sec:aux}).

For the analysis to follow, we also need to express the fact that the expected completion time of a job cannot be smaller than its expected processing time, which is generally not implied by \eqref{eq:LPy:compl-time}.
\begin{align}
C_j^{\ref{sto}}\geq \sum_{i\in M}\sum_{s\in\BZ_{\geq0}}y_{ijs}\qquad\text{for all~$j\in J$.}\label{eq:LPy:additional}	
\end{align}

The following LP relaxation for the unrelated machine scheduling problem can be derived with these observations.  This LP extends the LP  given in \cite{SSU2016} by adding the  constraints \eqref{eq:LPy:additional}.
\begin{align}
\tag{S}
\label{sto}
\begin{split}
\min\quad&z^{\ref{sto}}=\sum_{j\in J}w_j\,C_j^{\ref{sto}}\\
\text{s.t.}\quad&\text{\eqref{eq:LPy:mach-cap}, \eqref{eq:LPy:assignment}, \eqref{eq:LPy:compl-time}, \eqref{eq:LPy:additional}}\\
&y_{ijs}\geq0\qquad\qquad\text{for all~$j\in J$, $i\in M$, $s\in\BZ_{\geq0}$.}
\end{split}
\end{align}

The analysis in this paper will work with the dual of this relaxation. However the term $-\CV[\p_{ij}]^2$ in the primal objective would appear in the dual constrains. As we do not know how to deal with this negative term in the analysis that is to follow, we are going to factor it out.

To that end, define a simpler, \ie deterministic version for the expected completion times \eqref{eq:LPy:compl-time}, labeled with ``\ref{det}'' to distinguish it from the previous formulation, by letting
\begin{align}
\label{eq:LPy:simple-compl-time} C_j^{\ref{det}} = \sum_{i\in M}\sum_{s\in\BZ_{\geq 0}} \left(\frac{y_{ijs}}{\E[\p_{ij}]}\,\bigl(s+\tfrac12\bigr) +\frac{y_{ijs}}{2}\right) \qquad\text{for all~$j\in J$.} 
\end{align}
Consider the following linear programming problem
\begin{align}
\tag{P}
\label{det}
\begin{split}
\min\quad&z^{\ref{det}}=\sum_{j\in J}w_j\,C_j^{\ref{det}}\\
\text{s.t.}\quad&\text{\eqref{eq:LPy:mach-cap}, \eqref{eq:LPy:assignment}, \eqref{eq:LPy:simple-compl-time}}\\
&y_{ijs}\geq0\qquad\qquad\text{for all~$j\in J$, $i\in M$, $s\in\BZ_{\geq0}$\,.}
\end{split}
\end{align}
This corresponds to a time-indexed linear programming relaxation for a purely deterministic, unrelated machine scheduling problem where the random processing times are fixed at their expected values $\E[\p_{ij}]$. Also note that we have dropped constraints \eqref{eq:LPy:additional}.

In the following, a relationship between these two relaxations is established. To begin, define an upper bound on the squared coefficient of variation by
\begin{definition}
Define $\Delta$ as a universal upper bound on the squared coefficient of variation of the processing time of any job on any machine, that is 
\[
\Delta:=\max_{i,j} \CV[\p_{ij}]^2\,.
\]
\end{definition}
Observe that $\Delta=0$ for deterministic processing times, and $\Delta=1$ for processing times that are NBUE (new better than used in expectation), that is, the expected remaing processing time of a job never exceeds its total expected processing time. Specifically, $\Delta=1$ for exponential distributions.
Next, for any given solution $\boldsymbol{y}$ of \eqref{sto} or \eqref{det},  define
\[
H(\boldsymbol{y}):=\sum_{j\in J}w_j\sum_{i\in M}\sum_{s\in\BZ_{\geq0}}y_{ijs}\,.
\]
Let $\boldsymbol{y}^{\ref{sto}}$ denote an optimal solution to \eqref{sto} and recall that $\opt$ is the expected total weighted completion time of an optimal non-anticipatory scheduling policy.  By constraints \eqref{eq:LPy:additional},
\[
H(\boldsymbol{y}^{\ref{sto}}) = \sum_{j\in J}w_j\sum_{i\in M}\sum_{s\in\BZ_{\geq0}}{y}^{\ref{sto}}_{ijs}\le \sum_{j\in J} w_jC_j^{\ref{sto}}= z^{\ref{sto}}(\boldsymbol{y}^{\ref{sto}})\le \opt\,.
\]

The following lemma establishes the relation between the two relaxations and is crucial for our analysis.

\begin{lemma}
\label{lem:BoundlingLP-Det}
The optimal solution values $z^{\ref{det}}$ and $z^{\ref{sto}}$ of the linear programming relaxations \eqref{det} and \eqref{sto} fulfill
\[
z^{\ref{det}} \leq \bigl(1+\frac{\Delta}{2}\bigr)z^{\ref{sto}}\,.
\]
\end{lemma}
\begin{proof}{Proof.}
	 Let $\boldsymbol{y}^{\ref{det}}$ be an optimal solution to \eqref{det}
	and  $\boldsymbol{y}^{\ref{sto}}$ be an optimal solution to \eqref{sto}. Clearly,   $\boldsymbol{y}^{\ref{sto}}$ is a feasible solution also for \eqref{det} which is less constrained. Hence we get the following,  where $ z^{\ref{det}}(\boldsymbol{y}^{\ref{det}})$ is the value of $\boldsymbol{y}^{\ref{det}}$ on LP \eqref{det}.
	\begin{eqnarray}
	\begin{split}
	z^{\ref{det}}= z^{\ref{det}}(y^{\ref{det}}) &\le  z^{\ref{det}}(y^{\ref{sto}})\\
	&= z^{\ref{sto}}(\boldsymbol{y}^{\ref{sto}}) +\sum_{j\in J}w_j\sum_{i\in M}\sum_{s\in\BZ_{\geq0}}\frac{\CV[\p_{ij}]^2}{2} y_{ijs}^S\\
	&\le  z^{\ref{sto}}(\boldsymbol{y}^{\ref{sto}}) +\frac{\Delta}{2}H(\boldsymbol{y}^{\ref{sto}})\\
	&\le  \bigl(1+\frac{\Delta}{2}\bigr)z^{\ref{sto}}(\boldsymbol{y}^{\ref{sto}})\,.
	\end{split}
	\end{eqnarray}
	Note that the second-to-last inequality only uses the definitions of $\Delta$ and $H(\cdot)$. The last inequality holds because  $H(\boldsymbol{y}^{\ref{sto}})\le z^{\ref{sto}}(\boldsymbol{y}^{\ref{sto}})$. 
	\hfill \halmos
\end{proof}

Recalling that \eqref{sto} is a relaxation for the stochastic scheduling problem, we conclude the following.
\begin{corollary}
\label{cor:BoundingPrimal}
The optimal solution value $z^{\ref{det}}$ of the linear programming relaxation \eqref{det} is bounded by the expected performance of an optimal scheduling policy by
\[
z^{\ref{det}}\le \bigl(1+\frac{\Delta}{2}\bigr)\opt\,.
\]
\end{corollary}
The dual program of \eqref{det} will have unconstrained variables $\alpha_j$ for all $j\in J$ and nonnegative variables $\beta_{is}$ for all $i\in M$ and $s\in\BZ_{\geq0}$:
\begin{align}
\tag{D}
\label{dual}
\begin{split}
\max\quad &z^{\ref{dual}}  =\sum_{j\in J}\alpha_j\ -\ \sum_{i\in M}\sum_{s\in\BZ_{\geq0}} \beta_{is}\\
\text{s.t.} \quad& \frac{\alpha_j}{\E[\p_{ij}]} \le \beta_{is} + w_j\left(\frac{s+\frac12}{\E[\p_{ij}]} +\frac12\right) \text{ for all } i\in M,j\in J,s\in\BZ_{\geq0}\,,\\
&\beta_{is}\geq0 \hspace{25ex}\text{ for all } i\in M, s\in\BZ_{\geq0}\,.
\end{split}
\end{align}
Like the analysis in \cite{AnandGK12}, we will define a feasible solution for the dual \eqref{dual}, such that this solution corresponds to the schedule created by an online greedy algorithm for the original stochastic scheduling problem. Similar greedy algorithms have been used before, both in deterministic and stochastic scheduling on parallel machines, \eg\ in \cite{AA2003,ms2004,muv2006}. 

% !TEX root =  RevisionMOR2017-339.tex

\section{Greedy Algorithm \& Analysis for the Online-List Model.} \label{sec:greedy}
In this section the online-list model is considered. Assume without loss of generality that the jobs are presented in the order $1,2\dots,|J|$. 
On any machine $i$, let $H(j,i)$ denote the jobs that have priority no less than that of job $j$ according to the ratios $w_k/\E[\p_{ik}]$, breaking ties by index. That is, 
$$H(j,i):=\{k\in J\mid w_k/\E[\p_{ik}]>w_j/\E[\p_{ij}]\} \cup \{k\in J\mid k\le j, w_k/\E[\p_{ik}]=w_j/\E[\p_{ij}]\}.$$
Note that $j\in H(j,i)$. Also let $L(j,i):=J\setminus H(j,i)$.
Further let $k\to i$ denote  that a job $k$ has been assigned to machine $i$ by the algorithm.

\subsection*{Greedy Algorithm (Online List Model).} Whenever a new job $j\in J$ is presented to the algorithm,  compute for each of the machines $i\in M$ the \emph{instantaneous expected increase} in the cost if job $j$ is assigned to machine $i$,
 and all jobs already present on $i$ are scheduled in non-increasing order of the ratios weight over expected processing time. 
Since the expected completion time of the new job $j$ will be determined by the sum of expected processing times of all jobs in $H(j,i)$, and all the jobs in $L(j,i)$ will be delayed in expectation by an additional time  $\E[\p_{ij}]$, this cost increase equals
\[
\incr{j\to i}:=w_j\biggl(\sum_{k\to i, k \le j,k\in H(j,i)}\E[\p_{ik}]\biggr) \ +\  \E[\p_{ij}]\sum_{k\to i, k<j,k\in L(j,i)}w_k\,.
\]  
The greedy algorithm assigns the job to one of the machines where this quantity is minimal. That is, a job is assigned to machine $m(j):=\text{argmin}_{i\in M}\{\incr{j\to i}\}$; ties broken arbitrarily. Once all jobs have arrived and are assigned, the jobs assigned to a fixed machine are sequenced in non-increasing order of their ratio of weight over expected processing time. This WSEPT ordering is optimal conditioned on the given assignment \cite{Rothkopf66}.

The analysis of this greedy algorithm will proceed by defining a dual solution  $(\alpha,\beta)$ in  a way similar to that  done in \cite{AnandGK12}. Let
\[
\alpha_j:=\incr{j\to m(j)}\quad\text{for all }j\in J\,.
\]
That is, $\alpha_j$ is defined as the instantaneous expected increase in the total weighted completion time on the machine  job $j$ is assigned to by the greedy algorithm. Let
\[
\beta_{is}:=\sum_{j\in A_i(s)}w_j\,,
\]
where $A_i(s)$ is defined as the total set of jobs assigned to machine $i$ by the greedy algorithm, but restricted to those that have not yet been completed by time $s$ if the jobs' processing times were their expected values $\E[\p_{ij}]$. In other words, $\beta_{is}$ is exactly  the expected total weight of yet unfinished jobs on machine $i$ at time $s$, given the assignment (and sequencing) of the greedy algorithm.

It is now shown that these dual variables are feasible for the dual linear program. Later this fact will allow us to relate the variables to the optimal solution's objective.
\begin{lemma}
\label{fact:dual_feas}
The solution $(\boldsymbol{\alpha}/2,\boldsymbol{\beta}/2)$ is feasible for \eqref{dual}. 
\end{lemma}

\begin{proof}{Proof.}
This proof shows that
\begin{equation}
\label{eq:dual-feas}
\frac{\alpha_j}{\E[\p_{ij}]} \le  \beta_{is} + w_j\left(\frac{s}{\E[\p_{ij}]} +1 \right)\,
\end{equation}
holds for all $i\in M$, $j\in J$, and $s\in\BZ_{\geq0}$. This implies 
the feasibility of $(\boldsymbol{\alpha}/2,\boldsymbol{\beta}/2)$ for \eqref{dual}.
Fix  a job $j$ and machine $i$, and recall that $k\to i$ denotes a job $k$ being assigned to machine $i$ by the greedy algorithm.
By definition of $\alpha_j$ and by choice of $m(j)$ as the minimizer of $\incr{j\to i}$, for all $i$ it is the case that 
\begin{align}
\label{eq:roof_dual_feas1}
\frac{\alpha_j}{\E[\p_{ij}]}
&\le
\frac{\incr{j\to i}}{\E[\p_{ij}]}=w_j+w_j\!\!\!\!\!\!\!\sum_{k\to i, k< j,k\in H(j,i)}\frac{\E[\p_{ik}]}{\E[\p_{ij}]}\ +\ \sum_{k\to i, k<j,k\in L(j,i)}w_k\,.
\end{align}

Next, we are going to argue that the right-hand-side of \eqref{eq:roof_dual_feas1} is
upper bounded by the right-hand side of \eqref{eq:dual-feas}, from which the claim follows. Observe that the term $w_j$ cancels.
Observe that any job $k\to i$, $k\neq j$, can appear in the right-hand side of \eqref{eq:roof_dual_feas1} at most once, either with value $w_k$, namely when $k\in L(j,i)$, or with value $w_j\E[\p_{ik}]/\E[\p_{ij}]\leq w_k$ when $k\in H(j,i)$.  We show that each of these values in the right-hand-side of \eqref{eq:roof_dual_feas1} is
 accounted for in the right-hand side of \eqref{eq:dual-feas}, for any $s\geq 0$. 

Fix any such job $k\to i$. First consider the case that the time $s$ is small enough so that our job $k\to i$ is still alive at time $s$, so $s<\sum_{\ell\to i, \ell\in H(k,i)}\E[\p_{i\ell}]$. Then, $w_k$ is accounted for in the definition of $\beta_{is}$.

Now consider the case that $s\geq \sum_{\ell\to i, \ell\in H(k,i)}\E[\p_{i\ell}]$, which means that job $k$ is already finished at time $s$. In this case, we distinguish two cases.

\emph{Case 1} is $k\in L(j,i)$: In this case, job $k$ contributes to the right-hand side of \eqref{eq:roof_dual_feas1} a value of $w_k$, but as $s\geq \sum_{\ell\to i, \ell\in H(k,i)}\E[\p_{i\ell}]$,  the term $w_j({s}/{\E[\p_{ij}]})$ in the right-hand side of \eqref{eq:dual-feas} contains the term $w_j(\E[\p_{ik}]/\E[\p_{ij}])\geq w_k$.

\emph{Case 2} is $k\in H(j,i)$: In this case, job $k$ contributes to the right-hand side of \eqref{eq:roof_dual_feas1} a value of $w_j(\E[\p_{ik}]/\E[\p_{ij}])$, which is exactly what is also contained in the term $w_j({s}/{\E[\p_{ij}]})$, because $s\geq \sum_{\ell\to i, \ell\in H(k,i)}\E[\p_{i\ell}]$.\hfill\halmos
\end{proof}

In the following lemma, the online algorithm's objective is expressed in terms of the dual variables, which follows more or less directly from the definition of the dual variables $(\boldsymbol{\alpha},\boldsymbol{\beta})$. Let us denote by $\algo$ the total expected value achieved by the greedy algorithm. 

\begin{lemma}
\label{lem:dual_value}
The total expected value of the greedy algorithm is
\[
\algo = \sum_{j\in J}\alpha_j=\sum_{i\in M}\sum_{s\in\BZ_{\geq0}}\beta_{is}\,.
\]
\end{lemma}

\begin{proof}{Proof.}
For the first equality, recall that $\alpha_j$ is the  instantaneous increase in  $\algo$'s expected total weighted completion time. Summing this over all jobs gives exactly the total expected value of $\algo$'s objective. For a formal proof of this, see for example \cite[Lemma 4.1]{muv2006} for the case of parallel identical machines. That lemma and its proof can directly be applied to the case of unrelated machines. 

The second equality follows from the fact that the (expected) total weighted completion time of any schedule can be alternatively expressed by weighting each period of time by the total weight of yet unfinished jobs. The equality is true here, because $\boldsymbol{\beta}$ was defined on the basis of the same distribution of jobs over machines as given by $\algo$, and because each job $k$'s weight $w_k$, given $k\to i$, appears in $\beta_{is}$ for all times $s$ up to a job $k$'s expected completion time, given jobs' processing times are fixed to their expected values.
This is exactly what happens in computing the expected completion times under the greedy algorithm, because it is a ``fixed assignment'' algorithm that assigns all jobs to machines at time $0$, and sequences the jobs per machine thereafter. \hfill\halmos
\end{proof}

%\VG{
%From the proof of Fact~\ref{fact:dual_feas}, we can in fact deduce the following useful fact:
%\begin{fact}
%	\label{fact:relaxed_dual_feas}
%The solution $(\alpha,\beta)$ is feasible for a relaxed version of \eqref{dual} where the constraints are 
%\begin{align}
%	\label{eq:roof_dual_feas1_rel}
%\frac{\alpha_j}{\E[\p_{ij}]}-\beta_{is} & \le w_j\left(\frac{s+\frac12}{\E[\p_{ij}]} + 1 \right) \text{ for all } i\in M,j\in J,s\in\BZ_{\geq0}\,. 
%\end{align} 
%\end{fact}
%}

% !TEX root =  RevisionMOR2017-339.tex

\section{Speed Augmentation \& Analysis}
The previous analysis of the dual feasible solution $(\boldsymbol{\alpha}/2,\boldsymbol{\beta}/2)$ yields a dual objective value equal to $0$ by Lemma~\ref{lem:dual_value}. This is of little help to bound the algorithm's performance. However following \cite{AnandGK12}, define another dual solution which has an interpretation in the model where all machines run at faster speed $f\geq 1$, meaning in particular that all (expected) processing times get scaled (down) by a factor $f$.  

Define
$\algo^{f}$ as the expected solution value obtained by the same greedy algorithm, except that all the machines run at a speed increased by a factor of $f$, where $f\ge 1$ is an integer. Note that $\algo=f\algo^f$, by definition. We denote by $(\boldsymbol{\alpha}^f,\boldsymbol{\beta}^f)$ the exact same dual solution that was defined before, only for the new instance with faster machines. The following establishes feasibility of a slightly modified dual solution.

\begin{lemma}
\label{lem:new_dual_feas} 
Whenever $f\geq 2$, the solution $(\boldsymbol{\alpha}^f,\frac{1}{f}\boldsymbol{\beta}^f)$  is a feasible solution for the dual \eqref{dual} in the \emph{original} $($unscaled\,$)$ problem instance. 
\end{lemma}
\begin{proof}{Proof.}
By definition of $(\boldsymbol{\alpha}^f,\frac{1}{f}\boldsymbol{\beta}^f)$, to show feasibility  for \eqref{dual} it suffices to show the slightly stronger constraint that
\begin{align*}
\frac{\alpha_j^f}{\E[\p_{ij}]}
&\le
\frac{1}{f}\beta^f_{is}+w_j\left(\frac{s}{\E[\p_{ij}]}+\frac12\right)
\end{align*}
for all $i,j,s$. Indeed, in the above inequality we have only dropped the nonnegative term ${w_j}/{(2\E[\p_{ij}])}$ from the right-hand side of \eqref{dual}, hence the above implies the feasibility of $(\boldsymbol{\alpha}^f,\frac{1}{f}\boldsymbol{\beta}^f)$ for \eqref{dual}.
By definition of $\boldsymbol{\alpha}$ we have $\alpha_j=f\alpha_j^f$. So the above is equivalent to
\begin{equation}
\label{eq:proof_feas_new_dual1}
\frac{\alpha_j}{\E[\p_{ij}]}
 \le \beta^f_{is}+ w_j\left(\frac{f\cdot s}{\E[\p_{ij}]}+\frac{f}{2}\right)\,.
\end{equation}
As the assumption was that $f\geq 2$, \eqref{eq:proof_feas_new_dual1} is implied by
\begin{equation}
\label{eq:proof_feas_new_dual}
\frac{\alpha_j}{\E[\p_{ij}]}
 \le \beta^f_{is}+ w_j\left(\frac{f\cdot s}{\E[\p_{ij}]}+1\right)\,.
\end{equation}
Now observe that $\beta^f_{is}=\beta_{i(f\cdot s)}$ (and recall that $f$ is integer), so
 \eqref{eq:proof_feas_new_dual} is nothing but inequality \eqref{eq:dual-feas} with variable $s$ replaced by $f\cdot s$.  The validity of 
\eqref{eq:proof_feas_new_dual} therefore directly follows from \eqref{eq:dual-feas} in our earlier proof of Lemma~\ref{fact:dual_feas} to demonstrate the feasibility of $(\boldsymbol{\alpha}/2,\boldsymbol{\beta}/2)$ for \eqref{dual}. 
\hfill\halmos
\end{proof}

%\begin{lemma}
%\label{lem:new_dual_feas} 
%The solution $(\frac12{\alpha^f},\frac{1}{2f}\beta^f)$  is a feasible solution for the dual \eqref{dual} in the \emph{original} (unscaled) problem instance. 
%\end{lemma}
%\begin{proof}{Proof.}
%By definition of $(\frac12\alpha^f,\frac{1}{2f}\beta^f)$, to show feasibility  for \eqref{dual} it suffices to show
%\begin{align*}
%\frac{\alpha_j^f}{\E[\p_{ij}]}
%&\le
%\frac{1}{f}\beta^f_{is}+w_j\left(\frac{s}{\E[\p_{ij}]}+1\right)
%\end{align*}
%for all $i,j,s$. By definition of $\alpha$ we have $\alpha_j=f\alpha_j^f$. So the above is equivalent to
%\begin{equation}
%\label{eq:proof_feas_new_dual1}
%\frac{\alpha_j}{\E[\p_{ij}]}
% \le \beta^f_{is}+ w_j\left(\frac{f\cdot s}{\E[\p_{ij}]}+f\right)\,.
%\end{equation}
%As $f\geq 1$, \eqref{eq:proof_feas_new_dual1} is implied by
%\begin{equation}
%\label{eq:proof_feas_new_dual}
%\frac{\alpha_j}{\E[\p_{ij}]}
% \le \beta^f_{is}+ w_j\left(\frac{f\cdot s}{\E[\p_{ij}]}+1\right)\,.
%\end{equation}
%To see that this is true, observe that we can use the exact same argument as before in the proof of the feasibility of $(\alpha/2,\beta/2)$ in Fact~\ref{fact:dual_feas}, because the left-hand side of \eqref{eq:proof_feas_new_dual} is identical to the left-hand side of \eqref{eq:roof_dual_feas1}, and the right-hand side of \eqref{eq:proof_feas_new_dual} is identical to the right-had side of \eqref{eq:roof_dual_feas1}, because $\beta^f_{is}=\beta_{i(f\cdot s)}$.\hfill\halmos
%\end{proof}

\bigskip

The first main theorem of the paper is now established.
\begin{theorem}
\label{thm:no_release_dates}
The greedy algorithm has a performance guarantee of $(4+2\Delta)$ for online scheduling of stochastic jobs on unrelated machines to minimize the expectation of the total weighted completion times $\E[\sum_jw_jC_j]$.  That is, $\algo \le (4+2\Delta)\opt$.
\end{theorem}
\begin{proof}{Proof.}
We know from Corollary~\ref{cor:BoundingPrimal} that  $z^{\ref{dual}}(\boldsymbol{\alpha}^f,\frac{1}{f}\boldsymbol{\beta}^f)\le z^{\ref{dual}}=z^{\ref{det}}\le \bigl(1+\frac{\Delta}{2}\bigr)\opt\,$, given that $f\geq 2$.
Next, recall that $\algo^f=\sum_{j\in J}\alpha^f_j=\sum_{i\in M}\sum_{s\in\BZ_{\geq 0}}\beta^f_{is}$ by Lemma~\ref{lem:dual_value}, and $\algo=f\algo^f$.
 The theorem now follows from evaluating the objective value of the specifically chosen dual solution $(\boldsymbol{\alpha}^f,\frac{1}{f}\boldsymbol{\beta}^f)$ for \eqref{dual}, as
\[
z^{\ref{dual}}(\boldsymbol{\alpha}^f,\frac{1}{f}\boldsymbol{\beta}^f)=\sum_{j\in J}\alpha_j^f-\frac{1}{f}\sum_{i\in M}\sum_{s\in\BZ_{\geq0}}\beta^f_{is}=\frac{f-1}{f}\algo^f = \frac{f-1}{f^2}\algo\,.
\]
Putting together this equality with the previous inequality yields a performance bound equal to $\frac{f^2}{f-1}(1+\frac{\Delta}{2})$, where we have the constraint that $f\geq 2$. This term is minimal and equal to $(4+2\Delta)$, exactly when we choose $f=2$.
\hfill\halmos
\end{proof}

We end this section with the following theorem, which we believe was unknown before.
\begin{theorem}
The greedy algorithm for the deterministic online scheduling problem  has competitive ratio 4 for
minimizing the total weighted completion times $\sum_jw_jC_j$ on unrelated machines, and there is a tight lower bound of $4$ for the performance of the greedy algorithm.
\end{theorem}
\begin{proof}{Proof.}
The upper bounds follows as a special case of Theorem~\ref{thm:no_release_dates} as $\Delta=0$. As to the lower bound, we use a parametric instance from \cite{CorreaQueyranne2012}, which we briefly reproduce here for convenience. The instances are denoted $I^k$,  where $k\in \N$. There are $m$ machines, with $m$ defined large enough so that $m/h^2\in\N$ for all $h=1,\dots, k$. There are jobs $j=(h,\ell)$ for all $h=1,\dots,k$ and all $\ell=1,\dots,m/h^2$. The processing times of a job $j=(h,\ell)$
on a machine $i$ is defined as 
\[
p_{ij}=\begin{cases}1 & \text{ if } i\le \ell\,, \\ \infty & \text{ otherwise}\,.\end{cases}
\]
In other words, job $j=(h,\ell)$ can only be processed on machines $1,\dots,\ell$.  All jobs have weight $w_j=1$. As jobs have unit length on the machines on which they can be processed, we assume that the greedy algorithm breaks ties on each machine so that jobs with larger second index $\ell$ go first. 

The optimal schedule is to assign all jobs $j=(h,\ell)$ to machine $\ell$, resulting in $m/h^2$ jobs finishing at time $h$, for $h=1,\dots,k$, and hence a total cost $m\sum_{h=1}^k 1/h$. Now assume that the online sequence of jobs is by decreasing order of their second index. Then, as this is the same priority order as on each of the machines, the greedy algorithm assigns each job at the end of all previously assigned jobs. That means that the greedy algorithm assigns each job $j$ to one of the machines that minimizes its own completion time $C_j$. Here we assume that ties are broken in favour of lower machine index. It is shown in \cite{CorreaQueyranne2012} that the resulting schedule, which is in fact a Nash equilibrium in the game where jobs select a machine to minimize their own completion time, has a total cost at least $4m\sum_{i=1}^k1/i - \bigO{m}$. The lower bound of $4$ follows by letting $k\to\infty$.
\hfill\halmos
\end{proof}

\newcommand{\qxcomment}[1]{\textcolor{blue}{\textbf{[QX: #1]}}}

\section{The Online Time Model.}

This section addresses the online-time model where jobs arrive over time; that is, a job $j$ arrives at release time $r_j\ge 0$. In particular, the presence of job $j$ is unknown before time $r_j$. \red{\sout{Upon} At} time $r_j$, the job becomes available and the processing times distributions $\p_{ij}$ become known, for all machines $i=1,\dots,m$.
We may assume w.l.o.g.\ jobs are indexed such that $r_j\leq r_k$ for $j<k$. 

The difficulty in analyzing the problem where jobs arrive over time 
lies in jobs that block a machine for a long time, while shortly after, other jobs might be released that cannot be scheduled. This is a well known problem for the total weighted completion time objective in general, even for a single machine \cite{ms2004}. In order to counter that effect, a job $j$ is only started after an additional, forced delay that depends on its own expected processing time. For example for identical machine problems, \cite{ms2004} and \cite{muv2006} work with modified release times of the form $r_j':=\max\{r_j, c \E[\p_j]\}$, for \red{some parameter $c>0$}. Another idea to counter the same effect has been used in \cite{schulz2008}, namely to start a job no earlier than its (expected) starting time in a preemptive relaxation on a single machine that works $m$ times faster. For the unrelated machine problem that we consider here, we use  a combination of these two ideas. The assignment of jobs to machines will follow the same idea as for the case without release dates, namely to assign a job to a machine where (an approximation of) the expected increase of the objective value is minimal. Once assigned to a machine, for the stochastic case the modified release times will be defined on the basis of a ``nominal'' schedule where processing times $\p_{ij}$ are fixed at their expected values $\E[\p_{ij}]$. For that reason, this section first considers the deterministic problem where the processing times are defined by $p_{ij}:=\E[\p_{ij}]$ for all jobs $j$ and machines $i$.

\subsection{Nominal Schedule: Online Time Model with Deterministic Processing Times.}\label{sec:Alg_P}
Let us first describe the greedy algorithm that is used to assign jobs to machines and schedule jobs on machines. Per machine, it is actually the same greedy WSPT rule that prefers to schedule jobs with highest ratios weight over processing time $w_j/p_{ij}$, with the only difference that we also take into account modified release times. The assignment of jobs to machines is done greedily, too.
\subsubsection*{Greedy Algorithm (Online Time Model for Deterministic Processing Times).} 
Consider any fixed job $j$ that is is released at time $t=r_j$ with processing times $p_{ij}$ on machines $i=1,\dots,m$. Then we proceed as follows.
\begin{enumerate}
\item Define modified release times: On machine $i$ the release time of job $j$ is modified to $r_{ij} := \max\{ r_j, c \cdot p_{ij}\}$; we will optimize parameter $c$ later.
\item Let $U_i(t)$ denote the jobs which have been assigned to machine $i$ at time $t$ and that have not been started yet (excluding the fixed job $j$). %\red{Consider a hypothetical greedy WSPT schedule of jobs $U_i(r_j)$ with modified release time per machine $i$.}
%\red{\sout{Let $X_i(t)$ denote the remaining processing time of the job that is potentially being processed at time $t$ on machine $i$. If there is no such job, $X_i(t)=0$.}}
\item 
\red{
To decide on the assignment of job $j$ to a machine, we define $\costub{j\to i}$ as an upper bound on the additional cost of job $j$, when included into a hypothetical greedy WSPT schedule of jobs $U_i(r_j)$ on machine~$i$. 
%$\costub{j\to i}$ is an upper bound the cost difference of the two greedy WSPT schedules of jobs $U_i(r_j)$ on machine $i$, one including job~$j$, the other without job~$j$. 
The reason to work with an upper bound instead of the exact value, is potential jobs that could be released in the interval $(r_j,r_{ij})$. These could delay the earliest possible start time of job $j$ beyond $r_{ij}$. In defining $\costub{j\to i}$, we account for the maximum additional delay that such jobs could impose on $j$; see Lemma~\ref{lem:incr_r_j} below.
}
\item Among all machines $i\in\{1,\dots,m\}$, assign job $j$ to a machine $m(j)$ that minimizes $\costub{j\to i}$, ties broken arbitrarily.
\item\label{it:wspt} On each machine $i$, \red{\sout{we}} schedule jobs following the greedy weighted shortest processing time rule (WSPT) with modified release times $r_{ij}$. That is, as soon as a machine falls idle at time $t$, we schedule among all unscheduled jobs $k$ assigned to machine $i$ with $r_{ik}\le t$, any job $j$ with maximal ratio $w_k/p_{ik}$.
\end{enumerate}

\subsubsection*{Analysis.}
We now show that this greedy online algorithm is \red{7.216}-competitive. This is interesting in its own right because it improves on the best prior algorithm that was known to be $8$-competitive \cite{HSSW97}. 
As before, let us denote by $\algo$ the total value achieved by the greedy algorithm, and $\opt$ to be the optimal solution value.

\begin{definition}
\red{For job $j$ and machine $i$, define}
\[
\red{
\costub{j\to i}:= w_j  \left( \left( 1+\frac{1}{c} \right) r_{ij} +  p_{ij} + \sum_{k \in U_i(r_j) , \frac{w_k}{p_{ik}} \geq \frac{w_j}{p_{ij}}} p_{ik} \right) + \sum_{ {k \in U_i(r_j) , \frac{w_k}{p_{ik}} < \frac{w_j}{p_{ij}}} }  w_k p_{ij}\,.
}
\]
\end{definition}

\begin{lemma}\label{lem:incr_r_j}
\red{If $m(j)$ is the machine to which job $j$ got assigned by the greedy algorithm, then }
\red{
\begin{align*}
\algo  & \leq \sum_{j\in J}\costub{j\to m(j)}\,.
 \end{align*}
 }
 \end{lemma}
\begin{proof}{Proof.}
\red{Let $X_i(t)$ be the remaining processing time of a job that is in process on machine $i$ at time $t$, with $X_i(t)=0$ if no such job exists.
Consider a fixed job $j$'s contribution to $\sum_j w_j C_j$. When job $j$ is released at time $r_j$, it is assigned to a machine that minimizes $\costub{j\to i}$. We estimate the latest starting time of job $j$ on machine $i$, given the jobs $U_i(r_j)$ that have been assigned to the same machine: 
First, job $j$ can be started no earlier than time $r_{ij}$, and at time $r_{ij}$, the machine might be blocked for another $X_i(r_{ij})$ time units by some job $h$. Note that such job $h$ could even get released later than $r_j$, in time interval $(r_j,r_{ij})$.
Independent of this, $j$'s start can be further delayed by ``high priority jobs'' $k$ from $U_i(r_j)$, meaning that $w_k/p_{ik}\geq w_j/p_{ij}$. Finally, job $j$ could in turn delay the ``low priority jobs'' from $U_i(r_j)$.}
\red{
Hence, the increase of $\sum_j w_j C_j$, caused by job $j$ being assigned to machine $i$, is at most}
\[
\red{
w_j  \biggl( r_{ij}+ X_i(r_{ij}) + \sum_{k \in U_i(r_j) , \frac{w_k}{p_{ik}} \geq \frac{w_j}{p_{ij}}} p_{ik} + p_{ij} \biggr) + \sum_{ {k \in U_i(r_j) , \frac{w_k}{p_{ik}} < \frac{w_j}{p_{ij}}} }  w_k p_{ij}
\ \le\ \costub{j\to i}\,.
}
\]
\red{To see why the inequality is true, let $h$ be the potential job in process at time $r_{ij}$, then}
\[
\red{X_i(r_{ij}) \le p_{ih} \le \frac{r_{ih}}{c} \le \frac{r_{ij}}{c}\,.}
\]
\red{The claim now follows by summing over all jobs $j\in J$, and because of the following observation: In time interval $(r_j,r_{ij})$, even more ``high priority jobs'' $k$ could get released, and such jobs $k$ cause $j$'s start being delayed even further. But the delay that these jobs will impose on $j$, will be accounted for in the term $\costub{k\to i}$. The set of all ``low priority jobs'' that could get released in interval $(r_j,r_{ij})$, can cause $j$'s start being delayed by at most $X_i(r_{ij})$.
}
\hfill\halmos
\end{proof}
%Note the following. It may of course happen that some jobs of lower priority can get delayed by even more than $p_{ij}$, because the presence of job $j$ on machine $i$ can cause other higher priority jobs to become available later on, thereby causing low priority jobs to be delayed even further. But that increase will be accounted for at the time that such jobs appear; the total increase in total weighted completion time caused by job $j$ \emph{at time $r_j$} is bounded as above.

\begin{theorem}\label{thm:6}
The greedy algorithm for the deterministic online scheduling problem with release times has competitive ratio $\red{7.216}$ for
minimizing the total weighted completion times $\sum_jw_jC_j$ on unrelated machines. That is, $\algo\le \red{7.216}\,\opt$.
\end{theorem}
\begin{proof}{Proof.}
Let $m(\red{j})$ be the machine to which job $\red{j}$ got assigned. Define 
%$\alpha_j$ as the actual increase in total weighted completion time due to the presence of job $j$, and 
%$\beta_{i,s}$ as the sum of weights of unfinished jobs on machine $i$ at time $s$. 
%Moreover, denote by $\alpha_j$ the upper bound on job $j$'s contribution to the schedule computes by the greedy algorithm,
\begin{align*}
\alpha_j &:= \costub{j\to m(j)} \\
\beta_{i,s} &:= \sum_{k : m(k) = i; \ r_k \leq s; \ C_k \geq s} w_k
\end{align*}
\red{By definition of $\alpha$, $\beta$, and by Lemma~\ref{lem:incr_r_j} we then have}
\[ \algo =  \red{\sum_{i,s} \beta_{i,s} \le \sum_j \alpha_j}\,. \]
For this analysis, we again consider a speed scaled problem instance, but now we need to modify both the release times and the processing times by a factor $f$ as follows.
\begin{align*}
r_j^f &:= \frac{r_{j}}{f}\ \text{ and }\ 
p_{ij}^f := \frac{p_{\red{i}j}}{f}\,, \\
\intertext{so that we have}
r_{ij}^f & = \frac{r_{ij}}{f}\,.
\end{align*}
\red{Consider the same greedy algorithm on the scaled instance.} Observe that the machine assignment in the speed scaled instance is the same as in the original instance. In fact, the speed scaled instance just scales time by a factor of $f$. Define, $\alpha_j^f$ analogously as the \red{upper bound}  on the increase in total weighted completion time due to the presence of job $j$ in the speed scaled instance, and $\beta^f_{i,s}$ as the weight of the unfinished jobs on machine $i$ at time $s$ in the speed scaled instance. Then
\begin{align}
\begin{split}
\label{eq:def_alpha_beta_f}
\alpha_j^f &= \frac{\alpha_j}{f}, \\
\beta^f_{is} &= \beta_{i (f\cdot s)}.
\end{split}
\end{align}
(Here we assume w.l.o.g.\ that all job sizes \red{and release times} are integer multiples of $f$, which can be achieved by scaling.)
Also, let us denote by $\algo^f$ the value achieved by the greedy algorithm for the speed scaled instance, and note that $\algo^f= \sum_{i,s} \beta_{i,s}^f = \algo/f \red{\le \sum_j \alpha_j^f} $.

In the next section we are going to prove Lemma~\ref{lemma:dual_feasible} which gives a lower bound on the optimal solution value $\opt$, again via some feasible solution for the dual of a linear programming relaxation of the form $\left( \frac{\boldsymbol{\alpha}^f}{a} , \frac{\boldsymbol{\beta}^f}{b} \right)$ for  constants $(a,b)$, which will yield that
\begin{align*}
\opt \geq \sum_j \frac{\alpha^f_j}{a} - \sum_{i,s} \frac{\beta^f_{is}}{b} 
%= \algo^f \left( \frac{1}{a} - \frac{1}{b} \right) 
\ \red{\ge}\  \frac{\algo}{f} \left( \frac{1}{a} - \frac{1}{b} \right),
\end{align*}
or
\[ 
\algo \ \red{\leq}\  \frac{f\cdot\opt}{{1}/{a} - {1}/{b}}\,. 
\]
Now setting parameters \red{$c=2/3$, $a=32/23, b=16/3$, and speed $f = 23/6$} are feasible choices for using Lemma~\ref{lemma:dual_feasible}, which gives $\algo \leq \red{(7+11/51)\cdot \opt < 7.216} \cdot \opt$.\hfill\halmos
\end{proof}

\subsection{Linear Programming Relaxation and Dual Lower Bound.}

Analogous to the earlier linear programming relaxation \eqref{sto}, we can define the same LP relaxation for the  instance with release times $r_j$. We omit repeating this LP relaxation here as it is exactly the same as \eqref{sto}, except that the variables $y_{ijs}$ are defined only for times $s\geq r_j$. 
Let us refer to this modified LP relaxation for the problem with release dates $(\textup{S}_r)$, and its optimal solution value $z^{S_r}$. 
Similarly, analogous to (\ref{det}) we can define an LP relaxation for the deterministic version of the same problem with deterministic processing times $\E[\p_{ij}]$, by dropping all terms $-\CV[\p_{ij}]^2$ from the relaxation $(\textup{S}_r)$, and eliminating constraints 
\eqref{eq:LPy:additional}. Let us refer to this deterministic LP relaxation $(\textup{P}_r)$ with optimal solution value $z^{P_r}$. Lemma \ref{lem:BoundlingLP-Det} and Corollary \ref{cor:BoundingPrimal} apply to this linear programming relaxation in exactly the same way as before.
That is, when $\opt$ denotes the expected value of an optimal stochastic scheduling policy for the unrelated machine scheduling problem with release dates, we have that
\begin{equation}\label{eq:lp_bound_r_j}
z^{P_r}\leq \bigl(1+\frac{\Delta}{2}\bigr)z^{S_r}\,\leq \bigl(1+\frac{\Delta}{2}\bigr)\opt\,.
\end{equation}
Specifically, for the purpose of the proof of Theorem~\ref{thm:6}, observe that for the case of deterministic processing times where $\Delta=0$, the optimal LP solution value $z^{P_r}$ is simply a lower bound for $\opt$.

\medskip
\subsubsection*{Dual Lower Bound.} By duality, we can lower bound the optimal solution value $z^{P_r}$ for LP relaxation $(\textup{P}_r)$  by  any feasible solution to its dual linear program, which is:
\begin{align}
\tag{$\textup{D}_r$}
\label{dual_online}
\begin{split}
\max\quad &z^{D_r}  \ =\ \sum_{j\in J}\alpha_j\ -\ \sum_{i\in M}\sum_{s\in\BZ_{\geq0}} \beta_{is}\\
\text{s.t.} \quad& \frac{\alpha_j}{\E[\p_{ij}]}\ \le\  \beta_{is}\ +\ w_j\left(\frac{s+\frac12}{\E[\p_{ij}]} +\frac12\right) \text{ for all } i\in M,j\in J,s\in\BZ_{\geq r_j}\,,\\
&\beta_{is}\ \geq\ 0 \hspace{26.5ex}\text{ for all } i\in M, s\in\BZ_{\geq0}\,.
\end{split}
\end{align}

\begin{lemma} \label{lemma:dual_feasible}
With $\boldsymbol{\alpha}^f$ and $\boldsymbol{\beta}^f$ as defined in \eqref{eq:def_alpha_beta_f}, the values $(\frac{\boldsymbol{\alpha}^f}{a} , \frac{\boldsymbol{\beta}^f}{b} )$ are a feasible solution for the dual \eqref{dual_online}, given that \red{$af \geq 2(2+c)$, $1/c \le f(a-1)$, and $af\ge b$}.
Specifically for \red{$c=2/3$, $a=32/23, b=16/3$, and speed $f=23/6$}, the objective function value of the dual solution yields  $z^{D_r}(\frac{\boldsymbol{\alpha}^f}{a} , \frac{\boldsymbol{\beta}^f}{b})\ \red{\ge \ \frac{\algo}{f}(\frac{1}{a}-\frac{1}{b})\ =\ \frac{\algo}{7+11/51}}$.
\end{lemma}
\begin{proof}{Proof.}
We are only left to show the feasibility of the solution $( \frac{\boldsymbol{\alpha}^f}{a} , \frac{\boldsymbol{\beta}^f}{b} )$. For convenience, let us write $p_{ij}$ for $\E[\p_{ij}]$. Then the dual constraints require that, for all jobs $j$ and machines $i$, and for all times $s \geq r_j$
\begin{align}
\frac{\alpha_j}{p_{ij}} &\leq \beta_{is} + w_j \frac{s+\frac{1}{2}}{p_{ij}} + w_j \cdot \frac{1}{2}\,.
\end{align}
Let us fix job $j$ and machine $i$. Plugging in the values $\alpha^f_{j}/a$ and $\beta^f_{is}/b$, we need to show
\begin{align}
\frac{\alpha^f_j}{a\cdot p_{ij}} &\leq \frac{\beta^f_{is}}{b} + w_j \frac{s+\frac{1}{2}}{p_{ij}} + w_j \cdot \frac{1}{2}
\end{align}
for all $s\geq r_j$. Equivalently, noting that $\boldsymbol{\alpha}^f = \boldsymbol{\alpha}/f$, we have to show that
\begin{align}
\frac{\alpha_j}{p_{ij}} &\leq af \cdot \frac{\beta^f_{is}}{b} + w_j \frac{s+\frac{1}{2}}{p_{ij}}\cdot af + w_j \cdot \frac{af}{2} \,. 
\end{align}
Since $\beta^f_{is} = \beta_{i , fs}$ (the version with machines' speeds scaled by $f$ is just scaling down time by factor of $f$), and replacing $s+\frac12$ by $s$, it therefore suffices to show
\begin{align}
\label{eq:rhs}
\frac{\alpha_j}{p_{ij}} &\leq af \cdot \frac{\beta_{i,fs}}{b} + w_j \frac{s}{p_{ij}}\cdot af + w_j \cdot \frac{af}{2}  
\end{align}
for all $s\geq r_j$. Due to Lemma~\ref{lem:incr_r_j}, and our choice of $\alpha_j$ as minimizer of $\incr{j\to i}$ we have \red{for all machines $i$}  
\begin{align}
\label{eq:lhs}
\red{  \frac{\alpha_j}{p_{ij}}  \le  \frac{w_j}{p_{ij}} \cdot
  \left( \left( 1+\frac{1}{c} \right) r_{j} + p_{ij} + \sum\limits_{k \in U_i(r_j) , \frac{w_k}{p_{ik}} \geq \frac{w_j}{p_{ij}}} p_{ik} \right) + \sum\limits_{ {k \in U_i(r_j) , \frac{w_k}{p_{ik}} < \frac{w_j}{p_{ij}}}  } w_k \,.}
%= & \frac{w_j  \left( \left( 1+\frac{1}{c} \right) r_{j} + (1+c) P_{ij} + \sum_{k \in U_i(r_j) , \frac{w_k}{P_{ik}} \geq \frac{w_j}{P_{ij}}} P_{ik} \right) + \sum_{ {k \in U_i(r_j) , \frac{w_k}{P_{ik}} < \frac{w_j}{P_{ij}}} }  w_k P_{ij}}{P_{ij}} 
\end{align}
Hence it suffices to show that the right hand side in \eqref{eq:lhs}
is upper bounded by the right hand side in \eqref{eq:rhs}.
To that end, we even show a slightly stronger inequality is true: Recall that $\beta_{i,fs}$ is the total weight of jobs $k$ assigned to machine $i$ and unfinished at time $fs$
but with $r_k\le fs$. \red{As long as $f\ge 1$, and since $r_j\le s$}, we have
$r_j\le fs$.
Hence, $\beta_{i,fs} \ge \sum_{ k : m(k) = i, r_k \leq r_j, C_k \geq fs } w_k \ge \sum_{ k\in U_i(r_j), C_k \geq fs } w_k$.
Therefore it suffices to show that the right hand side of \eqref{eq:lhs}
is bounded from above by 
\begin{align*}
 &\frac{af}{b} \cdot  \sum_{ k\in U_i(r_j), C_k \geq fs  } w_k + w_j\frac{s}{p_{ij}} \cdot af + w_j \cdot \frac{af}{2} \\
 = &\left( \frac{af}{b} \cdot  \sum_{ k  \in U_i(r_j), C_k \geq fs } w_k + \frac{w_j}{p_{ij}} \cdot (fs-r_j)  \right) + \frac{w_j}{p_{ij}} \cdot (fs(a-1) + r_j ) + w_j \cdot \frac{af}{2}
\end{align*}
\red{Multiplying everything with $p_{ij}$}, we therefore need to argue that the following inequality is true
\begin{align*}
 & w_j  \left( \red{\left( 1+\frac{1}{c} \right) r_{ij} + p_{ij}} + \sum_{k \in U_i(r_j) , \frac{w_k}{p_{ik}} \geq \frac{w_j}{p_{ij}}} p_{ik} \right) + \sum_{ {k \in U_i(r_j) , \frac{w_k}{p_{ik}} < \frac{w_j}{p_{ij}}} }  w_k p_{ij} \\
& \qquad \qquad  \leq \left( \frac{af}{b} \cdot  \sum_{ k \in U_i(r_j)  :  C_k \geq fs } w_k p_{ij} + w_j \cdot (fs-r_j)  \right) + w_j \cdot (fs(a-1) + r_j ) + w_j{p_{ij}} \cdot \frac{af}{2}\,.
\end{align*}
Let us rewrite this more conveniently as
\red{
\begin{align}
\nonumber
 & w_j\cdot\underbrace{\biggl(\bigl( 1+\frac{1}{c}\bigr) r_{ij} + p_{ij}\biggr)}_{I}  
 \ +\  
 \underbrace{\sum_{k \in U_i(r_j) , \frac{w_k}{p_{ik}} \geq \frac{w_j}{p_{ij}}} w_j p_{ik}
 \ + \ 
 \sum_{ {k \in U_i(r_j) , \frac{w_k}{p_{ik}} < \frac{w_j}{p_{ij}}} }  w_k p_{ij}}_{II} \\
 \label{eq:final_step}
& \qquad \qquad  
\ \leq\ \  
\underbrace{ \frac{af}{b} \cdot  \sum_{ k  \in U_i(r_j) : C_k \geq fs } w_k p_{ij}
\ +\ 
w_j \cdot (fs-r_j)}_{II^*}  
\  + \ 
w_j \cdot \underbrace{\bigl((fs(a-1) + r_j ) + p_{ij}\cdot\frac{af}{2}\bigr)}_{I^*}\,.
\end{align}
}
The following observations and conditions are sufficient for the above inequality to be true.
\red{
\begin{enumerate}
\item  $I\le I^*$: Distinguish two cases. When $r_{ij}=r_j$, 
we have $I=(1+\frac{1}{c}) r_{ij} + p_{ij}= r_j+\frac{r_j}{c} + p_{ij}$. Moreover, since $s\ge r_j$,  $I^*=(fs(a-1) + r_j ) + p_{ij}\cdot\frac{af}{2}\ge r_j + f(a-1)r_j + p_{ij}\cdot\frac{af}{2}$. Therefore, we get that $I\le I^*$ under the conditions that  $1/c \le f(a-1)$, and $af\ge 2$.
On the other hand, when $r_{ij}=cp_{ij}$, we get $I=(2+c)p_{ij}$, and 
we get that $I\le I^*$ under the condition that $2(2+c)\le af$, whenever $a\ge 1$. Summarizing, we get that $I\le I^*$ for both cases, conditioned on  $1/c \le f(a-1)$ and $2(2+c)\le af$.
%\item Term II: Noting that $r_j \leq s$, it suffices that  $\frac{1}{c}\leq  f(a-1) $.
%\item Term III: It suffices that $ 2(1+c)\leq af$.
\item $II \le II^*$ : We have by definition of $U_i(r_j)$ that 
\[  w_j (fs-r_j) \ge w_j \sum_{ k \in U_i(r_j), C_k < fs } p_{ik} \,. \]
Therefore, under the condition that $\frac{af}{b} \geq 1$, we get that 
$II\le II^*$, because then
\[
II^*-II\ 
\ge\ \sum_{k \in U_i(r_j) , \frac{w_k}{p_{ik}} \geq \frac{w_j}{p_{ij}}, C_k\ge fs} (w_k p_{ij}- w_jp_{ik})
\ +\  
\sum_{k \in U_i(r_j) , \frac{w_k}{p_{ik}} < \frac{w_j}{p_{ij}}, C_k < fs} (w_j p_{ik}- w_kp_{ik})\ \ge 0\,.
\] 
\end{enumerate}
}
\hfill\halmos
\end{proof}
\subsection{Online Time Model with Stochastic Processing Times.}\label{sec:alg_S}

Let us first describe how we modify the greedy algorithm from the preceding section for the case with stochastic processing times.
\subsection*{Greedy Algorithm (Online Time Model with Stochastic Processing Times).}
When job sizes are stochastic, we use exactly the same greedy assignment of jobs to machines as we used in the preceding section for the deterministic case using processing times $p_{ij}:=\E[\p_{ij}]$. 

The only difference lies is the scheduling of jobs per machine, which works by restricting jobs to start no earlier than in the ``nominal'' schedule with deterministic processing times $p_{ij}=\E[\p_{ij}]$. Specifically, the jobs assigned to any machine $i$ are scheduled exactly in the same order as in the nominal schedule, with the $\ell$th job to start on machine $i$ starting at time
\[ S_{i,\ell} = \max\{  s_{i,\ell} , S_{i,\ell-1}+P_{i,\ell-1} \}\,. \]
Here, $s_{i,\ell}$ denotes the deterministic starting time of the $\ell$th job in the nominal schedule where $p_{ij}=\E[\p_{ij}]$ for all jobs $j$ and machines $i$. Here, note that the identity of the $\ell$th job to be scheduled on machine $i$ is the same in both cases. 
Also note that for the greedy algorithm for the stochastic case, the assignment of jobs to machines is deterministic, and not dependent on the realized processing times of jobs. The following two remarks are probably helpful.
\begin{enumerate}
\item[{\bf Remark 1.}] One may wonder if and how the algorithm can actually be executed online? This simply works by concurrently building the greedy WSPT schedule with deterministic processing times $p_{ij}:=\E[\p_{ij}]$. Consider any job $j$ that was released at time $r_j$. For the assignment of job $j$ to its correct machine $i=m(j$), indeed only information is needed that is available at time $r_j$. Also observe that it may be the case that neither the value $s_{i,j}$ is necessarily known at time $r_j$, nor which of the jobs are the predecessors of job $j$ on machine $i$. But this is not necessary, as job $j$ is simply blocked for processing as long as the same job has not started being processed in the corresponding deterministic schedule. 
\item[{\bf Remark 2.}] Observe that we may introduce forced idleness before the processing of any job~$j$. That is, even if the machine $i=m(j)$ is idle, we might not process any of the available jobs, 
and this delay depends on \red{the} nominal schedule for the underlying deterministic instance with $p_{ij}=\E[\p_{ij}]$. 
One may wonder why this forced idleness is actually necessary? Apart from the analysis that is to come, the reason to do that can most easily be seen by considering the following example: There are $n^2$ ``bad'' jobs of weight $\epsilon \ll 1$ released at time 0 with i.i.d.\ processing requirements $\p_{bad}=0$ with probability $1 - 1/n^2$ and  $\p_{bad}= n$ with probability $1/n^2$, and one  ``good'' job released at time 1 with weight 1 and deterministic processing time of 1. With the proposed algorithm that never starts a job before its starting time in the nominal schedule, we can schedule at most $n$ bad jobs  before the good job is released, because $\E[P_{bad}]=1/n$. That yields $\E[C_{good}]= \bigO{1}$. However without the this forced idle time, a greedy algorithm would keep scheduling bad jobs until there are none (if all are of size 0), or a rare long bad job is encountered. That yields $\E[C_{good}] =\Omega(n)$, which is problematic.
\end{enumerate}

\bigskip

The analysis of the greedy algorithm for the stochastic setting is based on a comparison with the nominal schedule, as expressed in the following lemma.
\begin{lemma}\label{lem:stoch_det}
The expected starting time of a job $j$ on machine $i$ in the stochastic case is bounded in terms of its starting time in the underlying nominal schedule\footnote{We write $S_{i,j}$ to indicate that job $j$ was assigned to machine $i$, only for notational convenience. As the assignment of jobs to machines is deterministic, observe that $S_j=S_{i,j}$.} by $\E[S_{i,j}] \leq h(\Delta)s_{i,j} $, where
\begin{align*}
h(\Delta) &= \begin{cases}
1+\frac{\sqrt{\Delta}}{2}, & \Delta \leq 1, \\
1+\frac{\Delta }{\Delta + 1}, & \Delta \geq 1\,.
\end{cases}
\end{align*}
\end{lemma}
Observe that $h(\,\cdot\,)$ is a concave, increasing function of $\Delta$, that $h(\Delta)\le 2$ for all $\Delta\ge 0$, and $h(0)=1$. Specifically, Lemma~\ref{lem:stoch_det} implies the weaker bound $\E[S_{i,j}] \leq 2s_{i,j}$.
\begin{proof}{Proof.}
For simplicity of notation, let us say that the jobs $k=1,\dots j$ are the jobs that have been assigned to machine $i$, in this order.
By definition of the algorithm for the stochastic setting, and by the fact that both the assignment to machines and the sequencing per machine is identical to the nominal schedule, the following equality holds per realization of the processing times. 
\begin{align*}
S_{i,j} & = \max \{ s_{i,j} , S_{i,j-1} + P_{i , j-1} \}  \\
&= \max \{ s_{i,j} , s_{i,j-1} + P_{i,j-1} , s_{i,j-2} + P_{i,j-2}+P_{i,j-1}, \ldots , P_{i,1} + \cdots + P_{i,j-1} \} \\
& =: F_{i,j}(P_{i,1}, P_{i,2}, \cdots, P_{i,j-1})
\end{align*}
Noting that the function $F_{i,j}$ is non-decreasing, Lipschitz continuous with coefficient 1 in each of its coordinates, and $F_{i,j}( \E[P_{i,1}], \cdots, \E[P_{i,j-1}] ) = s_{i,j} $, we have:
\begin{align*}
F_{i,j}(P_{i,1}, \cdots, P_{i,j-1}) & = 
F_{i,j}(\E[P_{i,1}], \cdots , \E[P_{i,j-1}]) 
+ 
\left( F_{i,j}(P_{i,1}, \cdots , P_{i,j-1}) 
- F_{i,j}(\E[P_{i,1}], \cdots , \E[P_{i,j-1}])  \right) \\
&  \leq s_{i,j} + \sum_{k=1}^{j-1} (P_{i,k} - \E[P_{i,k}])^+\,.
\intertext{Lemma~\ref{lemma:Extremal}, which is proved in the appendix, yields $\E[(P_{i,k} - \E[P_{i,k}])^+]\le (h(\Delta)-1)\E[\p_{i,k}]$. Hence taking expectations, we get }
\E[S_{i,j}] & \leq s_{i,j} + \sum_{k=1}^{j-1} \E[P_{i,k}] (h(\Delta)-1) \\
& \leq h(\Delta)s_{i,j}\,.
\end{align*}
The last inequality holds since for the nominal schedule, the $j$th job can not begin before time $\sum_{k=1}^{j-1} \E[P_{ik}]$, which means that $s_{i,j}\ge \sum_{k=1}^{j-1} \E[P_{ik}]$.\hfill\halmos
\end{proof}

%\begin{corollary}
%For each job $j$, under the greedy algorithm for the stochastic instance, the expected completion time of any job is at most twice the completion time of the same job in the schedule for the deterministic problem with 
%processing times $p_{ij}=\E[\p_{ij}]$. \end{corollary}

\bigskip

We conclude with the main theorem of this section.
\begin{theorem}
The greedy algorithm has a performance guarantee of \red{$(7.216+3.608\Delta)h(\Delta)$} for online scheduling of stochastic jobs with release times on unrelated machines to minimize the expectation of the total weighted completion times $\E[\sum_jw_jC_j]$.  That is, $\algo \le \red{(7.216+3.608\Delta)}h(\Delta) \opt$.
\end{theorem}

\begin{proof}{Proof.}
Let us denote by $C_j^P$ the completion time of job $j$ in the nominal schedule as computed by the greedy algorithm as described in Section~\ref{sec:Alg_P}, where $p_{ij}=\E[\p_{ij}]$. Let us denote by $\algo_P=\sum_{j}w_jC_j^P$ the objective value achieved by that nominal schedule.
Also, let us denote by $\algo=\sum_{j}w_j\E[C_j^S]$ the expected performance of the greedy algorithm for the stochastic case as described in this section.

It follows from Lemma~\ref{lem:stoch_det} that the expected completion time of any job $j$ under the greedy algorithm for the stochastic case fulfils
$\E[C_j]\leq h(\Delta)C_j^P$, and therefore
\[
\algo \le h(\Delta) \algo_P\,.
\]
What we have shown in Lemma~\ref{lemma:dual_feasible} is that there exists a solution to the dual LP relaxation $(\textup{D}_r)$ with value 
$\red{\ge \algo_P/(7+11/51)> \algo_P/7.216}$. Therefore by LP duality we get that $\algo_P \le \red{7.216}z^{P_r}$, with $z^{P_r}$ being the optimal solution value for the LP relaxation $(\textup{P}_r)$. That yields
\[
\algo \le h(\Delta)\algo_P \le h(\Delta) \red{7.216} z^{P_r}\le h(\Delta)\red{7.216}(1+\frac{\Delta}{2})z^{S_r}\le (\red{7.216+3.608}\Delta)h(\Delta)\opt\,.
\]
Here, the third inequality follows by \eqref{eq:lp_bound_r_j}.
\hfill\halmos
\end{proof}

\section{Conclusions}
The performance guarantees for the greedy algorithm obtained in this paper
are in the order $\bigO{\Delta}$, which is the same order of magnitude as earlier results that have been obtained for offline problems on unrelated machines \cite{SSU2016}, and of the same order of magnitude as earlier bounds for the online identical machines setting \cite{muv2006}. Getting results independent of $\Delta$ is an interesting open problem.

We also believe that the presented competitive analyses for the online deterministic problems are interesting in their own right, even if better competitive ratios can be obtained. We think so because the proposed greedy algorithm is (arguably) simple and intuitive, and hence practical. Finding a (matching) lower bound for the case with release times would be interesting.

Another direction for future work is the derivation of genuine lower bounds on the approximability of the optimal expected performance of efficiently computable policies for stochastic scheduling problems, even in the offline setting. This would allow to separate the computational complexity of stochastic problems from the corresponding deterministic special cases.

\subsubsection*{Acknowledgements.}
This work was started while all four authors were with the Simons Institute for the Theory of Computing at UC Berkeley.  The authors wish to thank the institute for the financial support and the organizers of the semester on ``Algorithms \& Uncertainty'' for providing a very stimulating atmosphere. A conference publication with preliminary results appeared in the proceedings of IPCO 2017 \cite{IPCO_paper}. B.\ Moseley was employed at Washington University in St.\ Louis while some of this research was conducted.  B.\ Moseley was supported in part by a Google Research Award, a Yahoo Research Award and NSF Grant CCF-1830711, 1824303, and 1733873. All authors wish to express their gratitude to the referees for their helpful comments and for challenging us to improve the results. Special thanks to Sven J\"ager from TU Berlin, for pointing out a flaw in one of the proofs in a previous version of this paper.

\bibliography{stochasticR}

\begin{APPENDICES}
% !TEX root =  RevisionMOR2017-339.tex

\section{Auxiliary Lemmas}
\label{sec:aux}

\begin{lemma}\label{lemma:LPy:compl-time}
We focus on a single machine and job. Let $P$ denote the random variable for the processing time with support $\BZ_{>0}$. Let $x_{t}$ denote the probability that the job starts processing on the machine at time $t$ $(t = 0,1,\ldots)$. For a given set of $\{x_{t}\}$ variables, let $y_{s}$ denote the probability that the job is being processed on the machine during time slot $s$. Then, the expected completion time of the job is given by
\[ C = \sum_{s\in\BZ_{\geq0}} \left(\frac{y_{s}}{\E[\p]}\,\bigl(s+\tfrac12\bigr)+\frac{1-\CV[\p]^2}{2}\,y_{s}\right). \]
\end{lemma}
\begin{proof}{Proof.}
It follows from the fact that policies are non-anticipatory that in terms of $\{x_t\}$ variables, the expected completion time is
\[ C =  \sum_{t=0}^{\infty} x_t (t + \E[P])\,. \]
Further, from \eqref{eq:fromxtoy},
\[ y_{s} = \sum_{t=0}^{s} x_t \cdot \Prb{P > s-t } , \]
which also gives
\[ \sum_{s=0}^\infty  y_s = \E[P] \sum_{t=0}^\infty x_t\,. \]
Consider the summation
\begin{align*}
\sum_{s=0}^{\infty} y_s \left( s + \frac{1}{2} \right) &= \sum_{s=0}^{\infty} \left( s+ \frac{1}{2} \right) \sum_{t=0}^s x_t \cdot \Prb{P > s-t} \\
&= \sum_{t=0}^{\infty} x_t \sum_{s=t}^\infty \left( s + \frac{1}{2} \right) \Prb{P> s-t}   \\
&= \sum_{t=0}^{\infty} x_t \left(  t \sum_{r=0}^\infty \Prb{P>r}  + \sum_{r=0}^{\infty} \left( r + \frac{1}{2} \right)  \Prb{P>r} \right) \\
&= \sum_{t=0}^{\infty} x_t \left(  t \cdot  \E[P]  +  \frac{1}{2} \E[P^2] \right) \\  
&= \E[P] \sum_{t=0}^{\infty} x_t \cdot t   +  \frac{1}{2} \E[P^2] \sum_{t=0}^{\infty} x_t \\  
&= \E[P] \left( \sum_{t=0}^{\infty} x_t \cdot t   +  \frac{1+\CV[P]^2}{2}  \sum_{s=0}^{\infty} y_s \right) 
\end{align*}
or,
\[ \sum_{t=0}^{\infty} x_t \cdot t  = \sum_{s=0}^{\infty} \left( \frac{y_s}{\E[P]} \left( s+ \frac{1}{2}\right) - \frac{1 + \CV[P]^2}{2} y_s  \right) . \]
Adding $ \sum_{t=0}^{\infty} x_t \E[P] = \sum_{s=0}^{\infty} y_s $ to the above, gives
\[ C = \sum_{t=0}^{\infty} x_t (t+\E[P]) = \sum_{s=0}^{\infty} \left( \frac{y_s}{\E[P]} \left( s+ \frac{1}{2}\right) + \frac{1-\CV[P]^2}{2} y_s  \right) . \]
\hfill\halmos
\end{proof}

\begin{lemma} \label{lemma:Extremal}
Let $X$ be a non-negative random variable with mean $\mu$ and squared coefficient of variation $\Delta$. Then,
\begin{align}
\E[ (X-\mu)^+ ] \leq \begin{cases}
\mu \frac{\sqrt{\Delta}}{2} & \Delta \leq 1, \\
\mu \frac{\Delta}{\Delta + 1} & \Delta \geq 1.
\end{cases}
\label{eqn:extremalbound}
\end{align}
\end{lemma}
\begin{proof}{Proof.}
We consider the problem of finding the measure on $\R_+$ for $X$ that maximizes $\E[(X-\mu)^+]$. Assuming $X$ has a density, this can be written as an infinite dimensional linear program. Using duality, we provide these bounds by exhibiting feasible solutions to its dual. It is instructive to follow this path, and later we argue that the bound holds for any distribution of $X$.
% While the formal presentation of infinite dimensional LPs and their duals is quite involved, for the sake of exposition we will begin by writing the LP and its dual at a somewhat non-rigorous level to motivate the structure of dual feasible solutions. At the end, we will use the dual feasible solution obtained to formally prove the stated bound without resorting to any infinite dimensional math. 
We note that the bounds are in fact tight which can be easily proved by analyzing the distributions in the duality-based proof, but we omit this step here. 

We begin with the primal problem of finding the extremal measure on $\R_+$ as an infinite dimensional LP, assuming the extremal measure has density $f(\cdot)$).
\begin{align*}
\max_{ f(\cdot) }  \quad &  \int_{0}^{\infty} (x-\mu)^+ f(x) dx \\
\mbox{s.t.} \quad & \int_{0}^{\infty} f(x) = 1 \\
& \int_{0}^\infty x \cdot f(x) dx = \mu \\
& \int_{0}^{\infty} x^2 \cdot f(x) dx = \mu^2 (1+\Delta).
\end{align*} 
The dual problem to the above is:
\begin{align*}
\min_{ \alpha, \beta, \gamma } \quad & \alpha + \mu \beta + \mu^2 (1+\Delta) \gamma \\
\mbox{s.t.} \quad & \alpha + \beta  x+ \gamma x^2 \geq (x-\mu)^+ \qquad \text{for all } x \geq 0.
\end{align*}

{\bf Case $\Delta \geq 1$}. We begin by making a guess about the extremal distribution $f(\cdot)$. Namely, that it is a parametric distribution with as limiting case two atoms, one of which is at $0$. Under this assumption, one observes that $\Prb{X=0} = \frac{\Delta}{\Delta + 1}$ and $\Prb{X= \mu(\Delta+1)} = \frac{1}{\Delta + 1}$ yields the desired conditions that $\E[X]=\mu$, $\E[X^2] = \mu^2(\Delta+1)$ and $\E[(X-\mu)^+] = \mu \frac{\Delta}{\Delta + 1}$. For the given guess, complementary slackness implies:
\begin{align}
\alpha + \beta \cdot 0 + \gamma \cdot 0 &= 0 \\
\alpha + \beta \cdot \mu (\Delta + 1) + \gamma \cdot \mu^2 (\Delta+1)^2 & = \mu \Delta.
\label{eqn:compslackness2}
\end{align}
The first of the above gives $\alpha = 0$, and the second gives the dual objective value of $\mu \frac{\Delta}{\Delta+1}$. It now remains to verify dual feasibility. Dual feasibility is met if the gradient of the quadratic function $\alpha + \beta x + \gamma x^2$ at $x=0$ is non-negative, and it is tangent to $(x-\mu)^+$ at $x=\mu (\Delta+1)$. The latter implies, 
\begin{align}
\beta + 2 \gamma \mu (\Delta+1) &= 1
\end{align}
which together with \eqref{eqn:compslackness2} gives, $\beta = \frac{\Delta - 1}{\Delta + 1}$ and $\gamma = \frac{1}{\mu (\Delta+1)^2}$. The non-negativity of gradient at $x=0$ is true if and only if $\beta \geq 0$. Therefore if $\Delta \geq 1$, then $\alpha = 0, \beta = \frac{\Delta-1}{\Delta+1}, \gamma = \frac{1}{\mu (\Delta+1)^2}$ is a feasible dual solution with objective value $\mu \frac{\Delta}{\Delta + 1}$. 

To turn this idea into a formal proof, we claim that
for any $x \geq 0$ and $\Delta \geq 1$, we have
\[ 
(x-\mu)^+ \leq \frac{\Delta-1}{\Delta+1} x + \frac{1}{\mu(\Delta+1)^2}x^2 \,.  
\]
This follows from $(x - \mu(\Delta+1))^2  \ge  0$ by adding $2x=(\Delta+1)x+(\Delta-1)x$ to both sides. Then basic algebra yields that $(x-\mu) \le (\Delta-1)x / (\Delta+1) + x^2 / (\mu(\Delta+1)^2)$, where we assume w.l.o.g.\ that $\mu>0$.  Therefore, 
for any random variable $X \geq 0$ with $\E[X]>0$ we have
\[ 
\E[(X-\mu)^+] \leq \frac{\Delta - 1}{\Delta+1} \E[{X}] + \frac{1}{\mu(\Delta+1)^2} \E[{X^2}] . \]
Substituting $\E[{X}]=\mu, \E[{X^2}] = \mu^2(1+\Delta)$, we get
\[ \E[{(X-\mu)^+}] \leq \mu \frac{\Delta}{\Delta+1}\,,  \]
which proves the second case of \eqref{eqn:extremalbound}.

{\bf Case: $\Delta \leq 1$.} The constraints of the dual suggest that we should look for a quadratic function $\alpha + \beta x + \gamma x^2$ which is tangent to $(x-\mu)^+$ at two points, one of which must be in the interval $[0,\mu]$. Therefore, $\alpha + \beta x + \gamma x^2 = \gamma (x - \nu_1)^2$ for some $0\leq \nu_1 \leq \mu$. Let this quadratic be tangent to $(x-\mu)$ at $\nu_2 \geq \mu$. The tangency conditions give
\begin{align}
\gamma (\nu_2 - \nu_1)^2 &= \nu_2 - \mu \\
2 \gamma (\nu_2 - \nu_1) &= 1
\end{align}
which together imply $\gamma = \frac{1}{4(\mu - \nu_1)}$. Thus, the dual minimization problem becomes the following single parameter optimization problem over $\nu_1$:
\begin{align*}
& \min_{\nu_1} \underbrace{\frac{\nu_1^2}{4(\mu-\nu_1)}}_{=: \alpha(\nu_1)}  + \mu \underbrace{\left( -\frac{\nu_1}{2(\mu-\nu_1)}  \right)}_{=: \beta(\nu_1)} + \mu^2 (\Delta+1) \underbrace{\frac{1}{4 (\mu-\nu_1)}}_{=:\gamma(\nu_1)}  \\
= & \min_{\nu_1} \frac{\mu^2 \Delta }{4(\mu-\nu_1))} + \frac{\mu-\nu_1}{4} \\
= & \frac{\mu \sqrt{\Delta}}{2}.  
\end{align*}
The minimizer is $\nu_1^* = \mu (1-\sqrt{\Delta}) $ which is indeed in the interval $[0, \mu]$ for $\Delta \leq 1$ as desired for dual feasibility.

As before, we can turn this into a formal proof by showing that for any $x \geq 0$, and $0 \leq \Delta \leq 1$,
\[ (x-\mu)^+ \leq    \frac{1}{4(\mu-\nu)} (x- \nu)^2\,, \]
where we have defined $\nu = \mu(1-\sqrt{\Delta})$. Then for any random variable $X \geq 0$ with $\E[X]>0$, 
\[ 
\E[(X-\mu)^+]  \leq \E\left[ \frac{1}{4(\mu-\nu)}(X-\nu)^2 \right]\,.
\]
After substituting $\E[{X}]=\mu, \E[{X^2}]=\mu^2(1+\Delta)$ this gives
\[ 
\E[{(X-\mu)^+}] \leq \mu \frac{\sqrt{\Delta}}{2}\,. 
\]
This completes the proof for the first case of \eqref{eqn:extremalbound}.

%\textbf{A formal proof of Lemma:} The proof follows with straightforward calculations given the dual solutions we have obtained above. 
%
%{\bf Case $\Delta \geq 1$:} Since for $x \geq 0$ and $\Delta \geq 1$
%\[ (x-\mu)^+ \leq \frac{\Delta-1}{\Delta+1} x + \frac{1}{\mu(\Delta+1)^2}x^2 \,,  \]
%for any random variable $X \geq 0$:
%\[ \E[(X-\mu)^+] \leq \frac{\Delta - 1}{\Delta+1} \E[{X}] + \frac{1}{\mu(\Delta+1)^2} \E[{X^2}] . \]
%Substituting $\E[{X}]=\mu, \E[{X^2}] = \mu^2(1+\Delta),$
%\[ \E[{(X-\mu)^+}] \leq \mu \frac{\Delta}{\Delta+1}.  \]
%
%{\bf Case $\Delta \leq 1$:} Since for $x \geq 0$, $0 \leq \Delta \leq 1$
%\[ (x-\mu)^+ \leq    \frac{1}{4(\mu-\nu)} (x- \nu)^2 \]
%for $\nu = \mu(1-\sqrt{\Delta})$, for any random variable $X \geq 0$
%\[ \E[(X-\mu)^+]  \leq \E\left[ \frac{1}{4(\mu-\nu)}(X-\nu)^2 \right] \]
%which after substituting $\E[{X}]=\mu, \E[{X^2}]=\mu^2(1+\Delta)$ gives:
%\[ \E[{(X-\mu)^+}] \leq \mu \frac{\sqrt{\Delta}}{2}. \]

{\bf Remark:} It was pointed out to us by one of the referees that the bound for the case $\Delta \leq 1$ follows even simpler by observing $\E[{(X-\mu)^+}] = \frac{1}{2} \E[{|X-\mu|_2}] \leq \frac{1}{2}\sqrt{\E[(X-\mu)^2]}=\frac{\mu\sqrt{\Delta}}{2}$. Here, the first equality follows because $X-\mu$ has mean zero, and the inequality is by Jensen's inequality. The linear programming argument however is constructive in the sense that it also gives tight examples.
\hfill\halmos

\end{proof}

\end{APPENDICES}

\end{document}